\documentclass[preprint]{elsarticle}
\usepackage{amssymb}
\usepackage{tikz}
\usepackage{amsmath}
\usepackage{lscape}
\begin{document}

\parindent 0pt

\parskip   0.3\baselineskip


\newcommand{\shrink}{\kern -0.75\baselineskip}

\newtheorem{theorem}{Theorem}
\newtheorem{lemma}[theorem]{Lemma}
\newtheorem{corollary}{Corollary}[theorem]
\newtheorem{proposition}{Proposition}
\newtheorem{property}{Property}
\newdefinition{rmk}{Remark}
\newdefinition{conj}{Conjecture}
\newdefinition{definition}{Definition}
\newdefinition{example}{Example}
\newproof{proof}{Proof}

\newcommand{\peo}{{\em peo}}
\usetikzlibrary{trees,positioning,arrows,shapes,snakes}


\begin{frontmatter}
\title{Integer Laplacian Eigenvalues \\ of Chordal Graphs}
\author{Nair Maria Maia de Abreu}
\address{ PEP-COPPE - Universidade Federal do Rio de Janeiro\\ 
Centro de Tecnologia, Bloco F, sala F105, Cidade Universit\'aria, 
  RJ, Brazil} 
  \ead{nairabreunovoa@gmail.com}
  \author{Claudia Marcela Justel\corref{cor1}}
\address{ Departamento de Engenharia de Computa\c c\~ao - Instituto Militar de Engenharia \\  Pra\c ca General Tib\'urcio 80, Praia Vermelha, 22290-270, 
  RJ, Brazil} 
  \ead{cjustel@ime.eb.br}
\author{Lilian Markenzon}
\address{ NCE - Universidade Federal do Rio de Janeiro\\ 
Av. Athos da Silveira, 274, Pr\'edio do CCMN,  Cidade Universit\'aria, 21941-611,  RJ, Brazil} 
 \ead{markenzon@ince.ufrj.br}
 \cortext[cor1]{Corresponding author}

\begin{abstract}
In this paper, structural properties of chordal graphs are analysed, in order to
establish a relationship between these structures and integer Laplacian eigenvalues.
We  present the characterization of chordal graphs with equal 
vertex and algebraic connectivities, by means of the vertices that compose the minimal vertex separators of the graph;
we stablish a sufficient condition for the cardinality of a maximal  clique to appear as an integer 
Laplacian eigenvalue.
Finally, we review two subclasses of chordal graphs, showing for them some new properties.
 \end{abstract}


\end{frontmatter}


\section{Introduction}\label{intro}

``Is it  possible, looking only at a graph structure, to determine its integer Laplacian eigenvalues?''
In this paper we address this question for chordal graphs.

Structural properties of  chordal graphs allow the development of efficient
solutions for many theoretical and algorithmic problems. 
In this context, its clique-based structure and  the minimal vertex
separators play a decisive role. 
Analysing these structures  we present new results about
integer Laplacian eigenvalues of chordal graphs and some subclasses.
Firstly, we characterize the chordal graphs with equal vertex and algebraic connectivities.
Then  we revisit  an integral Laplacian subclass, the quasi-threshold graphs \cite{YCC96}, and,
based on results from \cite{BLP08},
we show the relation between its structural and spectral properties.
The tree representation of a  quasi-threshold graph  is analysed,
allowing the determination of its Laplacian eigenvalues expressed in terms of maximal cliques and simplicial vertices.
Another interesting subclass is restudied: the $(k,t)$-split graphs \cite{F09,K10} and its properties raised.
Finally,  the importance of simplicial vertices in the existence of  integer Laplacian eigenvalues is established.


\section{Basic concepts}

Let  $G=(V,E)$ (or $G= (V(G), E(G))$)  be a connected graph, 
where $|E|=m$ is its {\em size} and
 $|V| = n $ is its {\em order}. 
 The {\em set of neighbors\/} of a vertex $v \in V$ is denoted by
$N(v) = \{ w \in V; \{v,w\} \in E\}$
 and its {\em closed neighborhood} by  $N_G[v] = N_G(v)\cup \{v\} $.
Two vertices $u$ and $v$   are  {\em false twins} in $G$ if  $N_G(u) = N_G (v)$, 
and  {\em true twins} in $G$  if $N_G[u] = N_G [v]$.
For any $S \subseteq V$, 
the subgraph of $G$ induced by $S$ is denoted $G[S]$. 
 If $G[S]$ is a complete subgraph then $S$ is a \emph{clique} in $G$. 
 The complete graph on $n$ vertices is denoted by $K_n$.
A vertex $v$ is said to be {\em
simplicial\/} in $G$ when $N(v)$ is a clique in $G$.

\bigskip

The Laplacian matrix of a graph $G$ of order $n$ is defined as $L(G) = D(G) - A(G)$, 
where $D(G) = diag(d_1,...,d_n)$ denote the diagonal degree matrix and $A(G)$ the adjacency matrix of $G$.
As $L(G)$ is symmetric, there are $n$ eigenvalues of $L(G)$. 
We denote the eigenvalues of $L(G)$, called the \emph{Laplacian eigenvalues}  of $G$, 
by $\mu_1(G)\geq \dots \geq \mu_n(G)$. 
Since $L(G)$ is positive semidefinite, $\mu_n(G)=0$ .
Fiedler \cite{F73} showed that $G$ is a connected graph if and only if $\mu_{n-1}(G) > 0$; 
this eigenvalue is called \emph{algebraic connectivity} and
it is denoted by $a(G)$. Moreover, Fidler proved that for $G \not = K_n$, $a(G) \leq \kappa(G)$.
All different  Laplacian eigenvalues of $G$ together with their multiplicities form the 
Laplacian spectrum  of $G$,  denoted by $Spec L(G)$. 
A graph is called \emph{Laplacian integral}  if its $L$-spectrum consists of integers.

\subsection{Chordal graphs}

A chordal graph is a graph in which every cycle of length four and greater has a cycle chord. 
Basic concepts about  chordal graphs are assumed to be known and 
can be found  Blair and Peyton \cite{BP93}  and Golumbic \cite{Go04}.
Following the most pertinent concepts are reviewed.

A subset $S \subset V$ is
a {\em separator} of $G$ if at least two vertices in the same connected
component of $G$ are in two distinct connected components of
$G[V\setminus S]$. 
The set $S$ is a {\em minimal separator} of $G$ if $S$ is a
separator and no proper set of $S$ separates the graph.
The {\em vertex connectivity} of $G$, $\kappa(G)$, is defined as the minimum cardinality of a separator of  $G$.

Let $G = (V, E)$ be a chordal graph and $u,v  \in V$. 
A subset $S \subset V$  is a {\em vertex separator}  for
non-adjacent vertices $u$  and $v$  (a $uv$-separator) if the
removal of $S$ from the graph separates $u$ and $v$  into distinct
connected components. 
If no proper subset of $S$  is a $uv$-separator then $S$ is a {\em minimal $uv$-separator}. 
When the pair of vertices remains unspecified, we refer to $S$  as a {\em
minimal vertex separator} ({\em mvs}). 
The set of minimal vertex separators is denoted by $\mathbb S$.
A minimal separator is always a minimal vertex separator
but the converse is not true.

A {\em clique-tree} of $G$ is defined as a tree $T$  whose vertices 
are the maximal cliques of $G$ such that for every  two maximal cliques $Q$ and $Q^\prime$
each clique in the path from $Q$ to $Q^\prime$ in $T$ contains $Q\cap Q^\prime$. 
The set of maximal cliques of $G$ is denoted by $\mathbb{Q}$.
For a chordal graph $G$ 
and a clique-tree $T=(\mathbb{Q}, E_T)$,
a set $S\subset V$ is a minimal vertex separator of $G$ if
and only if $S= Q\cap Q' $ for some edge $\{Q, Q'\}\in E_T$. 
Moreover, the multiset  ${\mathbb M}$ of
the minimal vertex separators of $G$ is the same for every
clique-tree of $G$.
The {\em multiplicity} of the minimal vertex separator $S$, denoted by
$\mu(S)$, is the number of times that $S$ appears in  ${\mathbb M}$. 
The set of minimal vertex separators  of $G$ is denoted by $\mathbb{S}$.
The determination of the minimal vertex separators and their multiplicities 
can be performed in linear time  \cite{MP10}.   

If $G_1=(V_1,E_1)$ and $G_2=(V_2,E_2)$ are graphs on disjoint set of vertices, 
their {\em graph sum} is $G_1 + G_2= (V_1 \cup V_2, E_1 \cup E_2).$
The {\em join} $G_1 \nabla G_2$ of $G_1$ and $G_2$ is a graph obtained from $G_1 + G_2$ 
by adding new edges from each vertex in $G_1$ to all vertices of $G_2$. 

 If $A$ is a finite $n$-element set then a {\em chain} is a collection of subsets 
$B_1, B_2,\ldots,B_k$   of $A$ such that for all $i,j\in \{1,2,...,k\}$ where 
$i\not= j$ we have that $B_i \subset B_j$    or   $B_j \subset B_i$.

\subsection{Some other graph classes}

Some other graph classes  will be mentionned in this paper.
Their definitions and characterizations can be found in \cite{BLS99}.

A graph $G$ is a {\em comparability graph}  if it transitively orientable, i.e. 
its edges can be directed such that if $a \rightarrow b$ and $b\rightarrow c$ are directed edges, then $a\rightarrow c$ is a directed edge. 

A graph $G$ is a {\em cograph} (short for complement-reducible graph) if one of the following equivalent conditions hold: 

\begin{itemize}
\item $G$ can be constructed from isolated vertices by disjoint union and join operations.
\item $G$ is $P_4$-free.
\end{itemize}

A \emph{split graph} is a graph $G=(V,E)$ if $V$ can be partitioned as the disjoint union of an
independent set and a clique. 
Split graphs are chordal graphs.

\section{Chordal graphs with $\kappa(G)=a(G)$}

In this section we characterize chordal graphs that have equal vertex and algebraic conectivities.
Our result particularizes an important result from Kirkland {\em et al.} \cite{K02},
allowing us to recognize this property based on the graph structure.

\begin{theorem}\label{theo:kirkland}{\rm\cite{K02}}
Let $G$ be a non-complete connected graph on $n$ vertices. 
Then $\kappa(G)=a(G)$ if and only if $G$ can be written as $G_1 \nabla G_2$, where $G_1$ is a disconnected graph
on $n - \kappa(G)$ vertices and $G_2$ is a graph on $\kappa(G)$ vertices with $a(G_2) \geq 2 \kappa(G) - n$.
\end{theorem}

\begin{theorem}\label{theo:chordal}
Let $G$ be a non-complete connected chordal graph.
Then $\kappa(G)=a(G)$ if and only if there is a minimal separator of $G$ such that

all its elements are universal vertices.
\end{theorem}
\begin{proof}
The vertex connectivity of a non-complete connected chordal graph is given
by the cardinality of its minimum separator.

Graph $G$ is connected, so, in order to obtain the disconnected graph $G_1$ stated by Theorem \ref{theo:kirkland},
 we  must remove at least a minimal separator $S$ of $G$.
  And more, as $G$ results from a join operation between $G_1$ and $G_2$,
 the vertices of $S$ must be adjacent to all vertices of $V(G_1)$.
 A minimal separator of  chordal graph $G$ is a minimal vertex separator of $G$.
As a minimal vertex separator is a clique, all the vertices of $V(G_2)$  are pairwise adjacent.
 Hence, they are universal vertices and $S$ is a minimum graph separator.
 
 As $G_2$ is a clique, $a(G_2) = |V(G_2)| = \kappa(G)$. 
 Then $\kappa(G) \geq 2 \kappa(G) - n = n \geq \kappa(G) $ is always true.
  \qed
\end{proof}

It is interesting to analyse the structure of a graph that satisfies Theorem \ref{theo:chordal}.
It is not difficult to see that the minimal vertex separators play an important role
in their characterization.
As the  minimal vertex separator $S=V(G_2)$ is composed by universal vertices, it must be a subset of
any other minimal vertex separator of the graph and, 
there is only one {\em mvs} with this property.

Figure \ref{fig:kappa} shows some examples. 
It is well known  that cographs obey this property.
Quasi-threshold graphs satisfy this property, since, as it will be reviewed in the next section,
they are cographs and chordal graphs. 
However there are chordal graphs that satisfy Theorem \ref{theo:chordal}  and are not cographs
as it is the case of the graph of Figure \ref{fig:kappa}(a).

\begin{figure}[h]
\begin{center}
\begin{tikzpicture}
  [scale=.27,auto=left]
\tikzstyle{every node}=[circle, draw, fill=white,
                         inner sep=0pt, minimum width=8pt]

\node (a) at (3,0){};
\node (b) at (2,4){};
\node (c) at (6,7) {};
\node (d) at (10,4){};
\node (e) at (9,0){};
\node (a1) at (1,-3){};
\node (b1) at (11,-3){};
\node (c1) at (6,-6){};
\node (mudo) at (6,-8) [draw=none,fill=none] {(a)};
\node (x) at (19,-3){};
\node (y) at (16,0){};
\node (z) at (22,0){};
\node (t) at (19,3){};
\node (s) at (22,6){};
\node (s1) at (25,3){};
\node (s2) at (25,-3){};
\node (mudo) at (21,-8) [draw=none,fill=none] {(b)};
\node (m) at (30,2){};
\node (n) at (33,5){};
\node (o) at (36,2){};
\node (p) at (33,-1){};
\node (q) at (31,-4){};
\node (q1) at (35,-4){};
\node (q2) at (37,-4){};
\node (q3) at (39,-4){};
\node (mudo) at (32,-8) [draw=none,fill=none] {(c)};

\foreach \from/\to in 
{a/b, a/c,a/d, a/e,b/d,b/c, b/e,c/d, c/e, d/e,
a1/a,a1/b,a1/e,
b1/a,b1/e,b1/d,
c1/a,c1/e,
x/y,x/z,x/t,y/z,y/t,z/t,t/s,s/z,s1/s,s1/z,
s1/s,s2/z,s1/t,
m/n,m/o,m/p,n/o,n/p,o/p,
q/m,q/p,
q1/p,q2/p,q3/p}
\draw (\from) -- (\to);
\end{tikzpicture}
\caption{Graphs with $a(G) = \kappa(G)$}
\label{fig:kappa}
\end{center}
\end{figure}
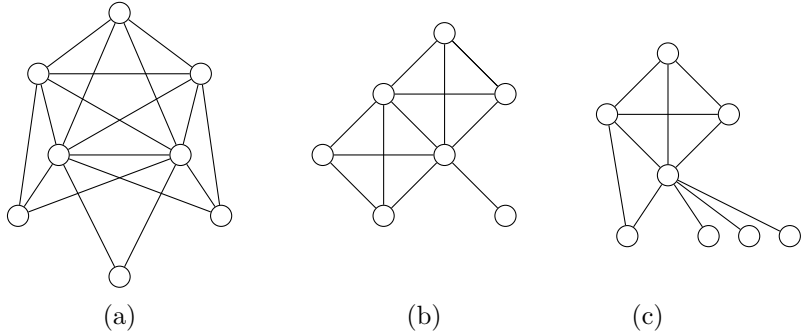

There is an efficient procedure to recognize if a graph $G$ obeys Theorem \ref{theo:chordal}.
First, the set ${\mathbb S}= \{ S_1, \ldots , S_r\}$ of
 minimal vertex separators must be determined; this step has $O(m)$ complexity time \cite{MP10}.
Then we must determine $\cap_{i=1}^r S_i$.
As $\sum_{i=1,r} |S_i| \leq m$ this step has also $O(m)$ complexity time.
It is immediate that $G$ has $\kappa(G)=a(G)$ if and only if this intersection
is for itself a minimal vertex separator  composed by universal vertices.



\section{Quasi-threshold graphs}

The {\em quasi-threshold graphs}  are the object of our study in this section.
An important subclass of chordal graphs, they
were defined in 1962 by Wolk \cite{W62} as {\em comparability graphs of a tree};
Golumbic \cite{G78} called them {\em trivially perfect graphs}.
Ma {\em et. al.} \cite{MWW89} called them {\em quasi-threshold graphs}
and studied algorithmic results.
Yan {\em et al.} \cite{YCC96} presented the following characterization theorem:

\begin{theorem}\label{theo:charact-quasi}
The following statements are equivalent for any graph $G$.
\begin{enumerate}
\item $G$ is a quasi-threshold graph.
\item $G$ is a cograph and is an interval graph.
\item $G$ is a cograph and is a chordal graph.
\item $G$ is $P_4$-free and $C_4$-free.
\item For any edge $uv$ in $G$, either $N[u] \subseteq N[v]$ or $N[v] \subseteq N[u]$.
\item If $v_1,v_2, \ldots, v_n$ is a path with $d(v_1) \geq d(v_2) \geq \ldots \geq d(v_{n-1})$, then
$\{v_1,v_2, \ldots, v_n\}$  is a clique.
\item G is induced by a rooted forest.
\end{enumerate}
\end{theorem}

It is immediate that, as a subclass of cographs, they are integral Laplacian graphs.

\subsection{Known subclasses}

It is interesting to review some subclasses of quasi-thresholds graphs already studied;
the results that will be presented for the quasi-threshold graphs will be obviously valid for all these classes.

\begin{itemize}
\item the {\em windmill graph} $Wd(k,\ell)$, $k \geq 2$ and $\ell \geq 2$, is a graph constructed   
by joining $\ell$ copies of a complete graph $K_k$ at a shared universal vertex. 
Figure \ref{fig:core}(a) presents $Wd(4,3)$.

\item the {\em split-complete graph}, which is the join of a  $K_k$ and a set of $n-k$ independent vertices.  
Figure \ref{fig:core}(b) presents an example with $7$  vertices.
\end{itemize}

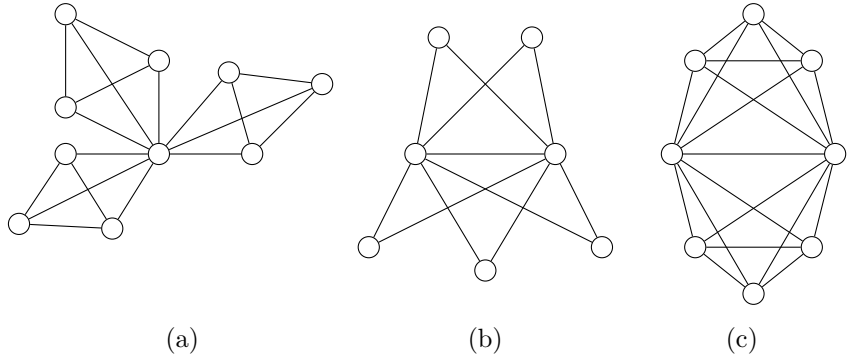
\begin{figure}[h]
\begin{center}
\begin{tikzpicture}
  [scale=.31,auto=left]
\tikzstyle{every node}=[circle, draw, fill=white,
                         inner sep=0pt, minimum width=8pt]

\node (x) at (-13,0){};
\node (x1) at (-15,-3){};
\node (x2) at (-11,-3.2){};
\node (t) at (-9,4){};
\node (t1) at (-13,2){};
\node (t2) at (-13,6){};
\node (s) at (-9,0){};
\node (w) at (-5,0){};
\node (w1) at (-6,3.5){};
\node (w2) at (-2,3){};
\node (mudo) at (-8,-8) [draw=none,fill=none] {(a)};
\node (a) at (2,0){};
\node (b) at (3,5){};
\node (c) at (7,5){};
\node (d) at (8,0){};
\node (a1) at (0,-4){};
\node (b1) at (10,-4){};
\node (c1) at (5,-5){};
\node (mudo) at (5,-8) [draw=none,fill=none] {(b)};
\node (o) at (13,0){}; 
\node (p) at (14,4){};
\node (p2) at (16.5,6){};
\node (q) at (19,4){};
\node (r) at (20,0){}; 
\node (p1) at (14,-4){};
\node (q1) at (19,-4){};
\node (p3) at (16.5,-6){};
\node (mudo) at (16,-8) [draw=none,fill=none] {(c)};

\foreach \from/\to in{x/x1,x/x2,x1/x2,x/s,x1/s,x2/s,
t/t1,t/t2,t1/t2,t/s,t1/s,t2/s,
w/w1,w/w2,w1/w2,w/s,w1/s,w2/s,
a/b,a/d,b/d,a/c,c/d,
a1/a,a1/d,b1/a,b1/d,c1/a,c1/d,
o/p,o/p2,o/q,o/r, p/p2,p/q,p/r, p2/q,p2/r, q/r,
o/p1, o/q1, o/p3, r/p1,r/q1,r/p3, p1/q1, p1/p3, q1/p3}
\draw (\from) -- (\to);


\end{tikzpicture}
\end{center}
\caption{Examples} 
\label{fig:core}
\end{figure}

Recently, a class that generalises windmill  and the split-complete graphs, 
the {\em core-satellite} graphs \cite{EB17},  were defined:

Let $c \geq 1$, $s \geq 1$ and $\eta \geq 2$. 
The {\em core-satellite graph}  is $\Theta (c,s,\eta) \cong K_c \nabla (\eta K_s)$. 
That is, they are the graphs consisting of $\eta $ copies of  $K_s$ (the satellites) 
meeting in a  $K_c$ (the core).
Figure \ref{fig:core}(c) presents a $\Theta (2,3,2)$ graph.
In the same paper, the {\em generalized core-satellite graphs} were defined, 
relaxing the size of the satellites in the previous definition;
this class is equivalent to a chordal graph that has only one minimal vertex separator and, as so,
it is a quasi-threshold graph.
Windmills have all maximal cliques with the same cardinality and 
the minimal vertex separator of cardinality 1;
split-complete graphs have all maximal cliques with cardinality $k+1$ and 
the minimal vertex separator of cardinality $k$; core-satellite graphs
have all maximal cliques with the same cardinality.

Another well-known  subclass is composed by the threshold graphs:
a graph is a {\em threshold graph} if it can constructed from the empty graph 
by repeatedly adding either an isolated vertex or an universal vertex. 
This construction can be expressed by a sequence of $0$'s  and $1$'s,
$0$ representing the addition of an isolated vertex, $1$ representing the addition of an universal vertex.

Threshold graphs are known to be the intersection of split  graphs  and cographs \cite{BLS99}.
In a connected threshold graph, the set of minimal vertex separators forms a chain (nested {\em mvs});
the {\em mvs} with smaller cardinality gives $\kappa(G)$ and, as this {\em mvs} is composed 
by universal vertices, it satisfies Theorem \ref{theo:chordal}.

\subsection{Structural $\times$ spectral properties}

{\em  Quasi-threshold graphs} are defined recursively by the following operations: 
\begin{itemize}
\item[$(\rho_1)$]  an isolated vertex is a quasi-threshold graph; 
\item[$(\rho_2)$] adding a new vertex adjacent to all vertices of a quasi-threshold graph results in a quasi-threshold graph;
\item[$(\rho_3)$] the disjoint union of two quasi-threshold graphs results in a quasi-threshold graph.
\end{itemize}

Yan {et al.} \cite{YCC96} presented a recognition algorithm for the class of quasi-threshold graphs
that, at the same time, build the rooted forest (or the tree, if the graph is connected) that induces the graph 
(item 7 of Theorem \ref{theo:charact-quasi}).
This directed tree is called a {\em tree representation} of $G$.
This algorithm is important in our quest for the integer Laplacian eigenvalues of the graph.
Given the graph $G=(V,E)$ the algorithm consists of two steps:
In Step 1 the directed graph $G_d =(V, E_d)$ is built.
An edge $\{v,w\}$ of $E$ becomes an arc $vw$ of $E_d$ if $d_G(v) \geq d_G(w)$.
In Step 2 the directed tree $T$ that induces $G$ is built.

Some structural properties of quasi-threshold graphs can be stated.

\begin{lemma}\label{lem:prop-quasi}
Let $G$ be a quasi-threshold graph.
The following properties hold:
\begin{itemize}
\item[(a)] Every maximal clique of $G$ has at least one simplicial vertex.
\item[(b)] Every maximal clique of $G$ contains at most one chain of minimal vertex separators.
\end{itemize}
\end{lemma}
\begin{proof}  The proof is by construction.

Operation $(\rho_3)$ joins two quasi-threshold graphs $G_1$ and  $G_2$, resulting
a new quasi-threshold graph $G$, not connected.
The operation does not affect the maximal cliques, the minimal vertex separators or
the simplicial vertices.

Operation $(\rho_1)$ establishes a new maximal clique composed by one simplicial vertex.
It is an isolated vertex and it is also a quasi-threshold graph with one simplicial vertex.

Operation $(\rho_2)$ adds a universal vertex $v$.
Three cases must be analysed:

\begin{itemize}
\item graph $G$ has only one maximal clique $Q$:
vertex $v$ is added to $Q$ and it becomes a simplicial vertex.

\item graph $G$ is connected and it has more than one maximal clique:
vertex $v$ is added to all maximal cliques. 
The simplicial vertices remain the same and the minimal vertex separators have 
an increment of one vertex (vertex $v$).

\item graph $G$ is not connected:
suppose that its components  are $C_1$, $C_2,\ldots ,$ $C_k$;
they are quasi-threshold graphs.
In this case, vertex $v$ establishes a new minimal vertex separator.
Vertex $v$ is added to all maximal cliques.
It belongs to all disjoint graphs and, as so, it stablishes a new $mvs$ that separates
vertices belonging to different components.
\end{itemize}

In all cases, a new maximal clique begins as a  simplicial vertex
that remains as such.
So, all maximal cliques have at least one simplicial vertex (item (a)).
All vertices added by operation $(\rho_2)$ are added to all minimal vertex separators
already existent. 
So, the new minimal vertex separator is contained in all the old ones
and  every maximal clique of $G$ contains at most one chain of minimal vertex separators (item (b)).
\qed
\end{proof}

\begin{lemma}\label{lem:prop-tree}
Let $G$ be a connected quasi-threshold graph and let $T$ be a tree representation of $G$
with root $r$.
Let $P= \langle r= v_1, \ldots, v_p\rangle$ be a maximal path in $T$.
Then
\begin{enumerate}
\item the vertices of $P$ establish a maximal clique $Q$ of $G$;
\item $\exists k, k \geq 2,$ such that $v_k, v_{k+1}, \ldots, v_p$ are simplicial vertices of $Q$.
\end{enumerate}
\end{lemma}
\begin{proof}  
An edge $\{u,v\}$ of $G$ is an arc $uv$ of $G_d$ if
$d_G(u) \geq d_G(v)$.
So, by Theorem \ref{theo:charact-quasi},
 it is immediate that the vertices of $P= \langle r= v_1, \ldots, v_p\rangle$
establish a clique in $G$. 
As it is a maximal path, it is  a maximal clique.
In a maximal clique, the simplicial vertices have the smaller degree, since
they are linked with only the vertices of that maximal clique.
So, all simplicial vertices apppear in the end of the path.
By Lemma \ref{lem:prop-quasi} the non-simplicial vertices compose
one chain of minimal vertex separators; they form the first part of the path.

Observe that all leaves in $T$ are simplicial vertices, since all
maximal cliques have at least one simplicial vertex (Lemma \ref{lem:prop-quasi}).
\qed
\end{proof}

Bapat {\em et al.} \cite{BLP08} have presented the integer Laplacian eigenvalues of the class 
of weakly quasi-theshold, which generalises the quasi-threshold graphs.
Their result provides the integer Laplacian eigenvalues of a quasi-threshold graph in terms  of 
the directed tree $T$:

\begin{theorem}{\em (Corollary 2.2 \cite{BLP08})}\label{theo:quasi-Bapat}
Let $G$ be a connected quasi-threshold graph. 
Suppose $G$ is the comparability graph of a rooted tree $T=(V,E_T)$ with vertices $u_1, u_2,\ldots, u_k$. 
Then the nonzero  Laplacian eigenvalues of $G$ are
\begin{itemize}
\item[\em(1)] $d_G(u_i) + 1$, repeated exactly once for each non-pendant vertex $u_i$, and
\item[\em(2)]  $m_i + 1$, repeated $c_i - 1$ times for each non-pendant vertex $u_i$,

being  $m_i$, $2 \leq i \leq k$,  the distance of the vertex $u_i$ from the root vertex $u_1$
and being $c_i=|child(u_i)|$ where $child(u_i)$ is the set of vertices $v$ such that there is an arc $u_i v \in E_T$.

\end{itemize}
\end{theorem}

The same result can be expressed by structural aspects of chordal graphs.

\begin{theorem}\label{theo:quasi-structural}
 Let $G$ be a  connected quasi-threshold graph. 
Let $\mathbb Q$ be its set of maximal cliques, $\mathbb S$ its set of
minimal vertex separators and $Simp(Q_i)$ the set of simplicial vertices of $Q_i, \forall Q_i\in \mathbb Q$. 
  Then the nonzero Laplacian eigenvalues of $G$ are
  \begin{itemize}
\item[\em(a)]  $d_G(v) + 1$, repeated exactly once for each non-simplicial vertex $v$, and
\item[\em(b)] $|Q_i|$, repeated $Simp(Q_i)-1$ times, $\forall Q_i \in \mathbb Q$, and
\item[\em(c)] $|S_i|$ repeated  $\mu(S_i)$ times.
\end{itemize}
\end{theorem}
\begin{proof}  
Let us consider the tree representation of $G$, which is the rooted tree $T$, as seen in Theorem \ref{theo:quasi-Bapat}.
We must analyse each item of Theorem \ref{theo:quasi-Bapat}.

\begin{itemize}
\item[(1)] $d_G(u_i) + 1$, repeated exactly once for each non-pendant vertex $u_i$.
\end{itemize}

A non-pendant  vertex of $T$ can be a simplicial vertex or a vertex belonging to a minimal vertex separator.
Firstly we are going to consider the non-simplicial vertices, rewriting partially item (1).

\begin{itemize}
\item[(a)] $d_G(v) + 1$, repeated exactly once for each non-simplicial vertex $v$.
\end{itemize}

We know, by Lemma \ref{lem:prop-tree}, that all leaves of $T$ are simplicial vertices and each maximal path 
of $T$ corresponds to a maximal clique.
So, if the maximal clique have only one simplicial, it is the leaf of the path; if there is more than one simplicial,
 it is a non-pendant vertex.
So, for each maximal clique $Q_i$ there is $Simp(Q_i) - 1$ simplicial vertices that are not-pendant vertices.
These vertices have degree $|Q_i| -1$. Item (1) can be completed:

\begin{itemize}
\item[(b)] $|Q_i|$, repeated $Simp(Q_i)-1$ times, $\forall Q_i \in \mathbb Q$.
\end{itemize}

We now consider item (2) of Theorem \ref{theo:quasi-Bapat}.
The value $m_i, 2 \leq i \leq k$, corresponds to  the distance of the vertex $u_i$ from the root vertex $u_1$ of $T$.
As it is repeated $c_i-1$ times, only vertices that have more than one child in $T$ must be observed;
only vertices belonging to minimal vertex separators obey this condition since this means that the vertex
belongs to more than on maximal clique.

Each $mvs$  is represented by a path of the tree $T$.
Let $\langle r, v_2, \ldots, v_p\rangle$ be a path such that $child(v_p)>1$ ($v_p$ is a non-pendant vertex).
Suppose, without loss of generality, that $v_p$ has two children, $a$ and $b$.
There is at least two maximal cliques in the graph: $\{ r, v_2, \ldots, v_p, a,\ldots\}$
and $\{ r, v_2, \ldots, v_p, b, \ldots\}$; $\{ r, v_2, \ldots, v_p\}$ is a $mvs$ $S$.
By Theorem \ref{theo:quasi-Bapat}, $|\{ r, v_2, \ldots, v_p\}| = |S|$ appears as an integer Laplacian eigenvalue,
since it corresponds to the distance of $v_p$ from $r$ plus 1.
The number of children of vertex $v_p$ corresponds to the multiplicity of the $mvs$ plus 1.
So, item (2) is equivalent to:
\begin{itemize}
\item[(c)]  $|S_i|$ repeated  $\mu(S_i)$ times. \qed
\end{itemize}
\end{proof}

\begin{theorem}
The determination of the integer Laplacian eigenvalues of a quasi-threshold graph has linear time complexity.
\end{theorem}
\begin{proof}
The determination of the degrees of the vertices of any graph takes $O(m)$ complexity time.
For a chordal graph, the determination of maximal cliques and minimal vertex separators
has linear time complexity \cite{MP10}.
Since a simplicial vertex is a vertex that appears in only maximal clique, 
their  determination also has $O(m)$ complexity time. 
So, the results presented in Theorem \ref{theo:quasi-structural} can be
determined in linear time complexity.
\end{proof}


\section{$(k,t)$-split graphs}

In this section we approach a subclass of split graphs, the $(k,t)$-split graphs, which 
were studied by Freitas \cite{F09} and Kirkland {\em et al.} \cite{K10}.
Our results establish a new approach of the results already presented.
It is known that $(k,t)$-split  graphs have, at least, some integer Laplacian eigenvalues and 
we explicitly determine these values  in terms of the graph structure. 

As already seen,  $G=(V,E)$ is a {\em split graph} if $V$ can be partitioned in an independent set and a clique. 
 Such partition may not be unique and it is referred as a {\it split partition} of $V$. 
 However, every connected bipartite split graph has a split partition equals to the degree partition 
 which is called  {\em split degree partition} (or {\em sdp}).  
 A {\em split-complete graph} is a split graph such that each vertex of the independent set is adjacent to every vertex of the clique. 

Let $H$ be a bipartite graph were the degree partition equals the bipartite partition. 
So, $H$ is a regular or a biregular graph. In the last case, $H$ will be called a {\em genuine biregular graph}. 
Observe that there are  biregular bipartite graphs which are not  genuine biregular graphs; for instance, $P_4$.
According to Kirkland {\em et al.} \cite{K10}, every connected  biregular split graph $G$ 
has a split partition which is equal to its degree partition. 
Moreover, $G$ has a unique regular or a unique genuine biregular graph $H$ as a spanning subgraph
called the {\it splitness} of $G$. 
Note that $H$ can be obtained from $G$ taking off all edges in the clique determined by its $sdp$.

 Let $G$ be a split graph, $\mathbb S = \{S_i: 1 \leq i \leq r \}$  its set of minimal vertex separators and 
 $T$-$Simp$ its set of simplicial vertices. 
Graph $G$ is a {\it $(k,t)$-split graph} if and only if the following properties hold:
\begin{itemize}
\item $k= |S_i|, \forall S_i \in \mathbb S$;
\item $S_i \cap S_j = \emptyset, \forall S_i, S_j \in \mathbb S$;
\item $t=|A_i|$ being $A_i$ the set of false twins adjacent to vertices of $S_i$, $1\leq i \leq r$;
\item $\bigcup_{i=1,r} A_i = T$-$Simp$.
\end{itemize}

Figure \ref{fig:split} displays a $(2, 3)$-split graph. In this case, $r=3$. 

If $G$ is a $(k,t)$-split graph, the following statements are valid:

\begin{enumerate}
\item $G$ is a biregular graph.
\item $n=(k + t)r$.
\item For $r >1$, $G$ has a unique split partition. So the split partition equals the degree partition.
\item The splitness $H$ of   $G$ can be obtained taking off 
all edges of the clique  determined by its non-simplicial vertices. 
\item The splitness $H$ of $G$ is isomorphic to $r$ copies of $K_{k,t}$, that is, $H = rK_{k,t}$. 
\item  $G$ is a split-complete graph if and only if $G$ is a $(k,t)$-split graph with $|{\mathbb S}|=1$.
\end{enumerate}

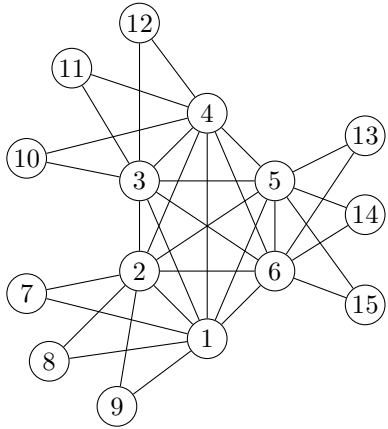
\begin{figure}[h]
\begin{center}
\begin{tikzpicture}
  [scale=.3,auto=left]
\tikzstyle{every node}=[circle, draw, fill=white,
                         inner sep=0pt, minimum width=15pt]
\node (a) at (7,-3){$1$};
\node (b) at (4,0){$2$};
\node (f) at (10,0){$6$};
\node (d) at (7,7) {$4$};
\node (c) at (4,4){$3$};
\node (e) at (10,4){$5$};
\node (g) at (-1,-1){$7$};
\node (h) at (0,-4){$8$};
\node (i) at (3,-6){$9$};
\node (j) at (-1,5){$10$};
\node (k) at (1,9){$11$};
\node (l) at (4,11){$12$};
\node (m) at (14,6){$13$};
\node (n) at (14,2.5){$14$};
\node (o) at (14,-1.5){$15$};
\foreach \from/\to in {a/b, a/c,a/d, a/e,a/f, b/d,b/c, b/e,b/f,c/d, c/e, c/f, d/e,d/f,e/f,  
g/a,g/b,h/a,h/b,i/a,i/b,j/c,j/d,k/c,k/d,l/c,l/d,m/e,m/f,n/e,n/f,o/e,o/f}  
\draw (\from) -- (\to);
\end{tikzpicture}
\caption{A $(2,3)$-split graph}
\label{fig:split}
\end{center}
\end{figure}

We are going to consider a particular labeling of the vertices of a $(k,t)$-split graph, that will be used in the construction
of its Laplacian matrix.
We begin by labeling the vertices of each minimal vertex separator $S_i, 1 \leq i \leq r$.
Then we label each set $A_i,1 \leq i \leq r$.
For instance, the graph depicted in Figure \ref{fig:split} was already labeled in this way.

The Laplacian matrix of a $(k,t)$-split graph $G$ has the form:

$$
 L(G)=
\begin{vmatrix}
(rk+t) \mathbb{I}_{rk} - \mathbb{J}_{rk} & - \mathbb{X}_{rk\times rt} \\
-\mathbb{X}^T_{rt\times rk} & k\mathbb{I}_{rt}\\
\end{vmatrix}
(1)
 $$
 
where $\mathbb{X}$ is a block diagonal matrix with $r$ blocks $\mathbb{J}_{k,t}$. 

The next theorem 
determines the  integer Laplacian eigenvalues of a $(k,t)$-split graph in terms of its structure.
If $r=1$, $G$ is a split-complete graph and, as so, it is a Laplacian integral graph. 
When $r \neq 1$, the proof follows straightforward from Theorem 2.2 in \cite{K10}. 

\begin{theorem}\label{theo:k-t-split}
Let $G$ be a $(k,t)$-split graph and ${\mathbb S} = \{S_1, \ldots, S_r\}$ 
its set of minimal vertex separators.
Then $G$ has, at least, $n-2r+2$ integer Laplacian eigenvalues.
\end{theorem}
\begin{proof}
If $r =1$, $G$ is a split-complete graph with $t=n-k$ vertices in its independent set. 
Then the Laplacian eigenvalues follow from Theorem \ref{theo:quasi-structural}. 
The graph has the following eigenvalues: 
$0$ with multiplicity $1$, $k$ with multiplicity $t -1$ and $k+t$ with multiplicity $k$. 
If $r \neq 1$, the splitness of $G$ is $H=rK_{k,t}$. 
According to  (1), 
the rank of block $\mathbb{X}$ of matrix $L(G)$ is $r$, the number of copies of $K_{k,t}$. 
Then by Theorem 2.2 \cite{K10}  the following integer Laplacian eigenvalues of $G$ are determined: 
$0$ with multiplicity $1$, $k$ with multiplicity $r(t-1)$; $rk+t$ with multiplicity $r(k-1)$ and $k+t$ with multiplicity $1$.  

As $n= (k+t)r$, then $r(t -1) +  r(k +1) +2 = n-2r+2$. 
\qed
\end{proof}

Form Theorem \ref{theo:k-t-split}, it is immediate to see that a $(k,t)$-split graph $G$ has, 
at least, the following integer Laplacian eigenvalues:
\begin{itemize}
\item[--] $0$ with multiplicity $1$,
\item[--] $k$ with multiplicity $r(t - 1)$,
\item[--] $rk + t$ with multiplicity $r(k - 1)$,
\item[--] $k + t$ with multiplicity $1$.
\end{itemize}

Observe that  $(k,t)$-split graphs have an interesting property:
for $r=1$, the graph is a split-complete graph and $a(G) = \kappa(G)$;
for $r >1$, $a(G)$ is always different from $\kappa(G)$ since, in this case,
there is not a minimal vertex separator composed by universal vertices (Theorem \ref{theo:chordal}).


\section{Simplicial vertices and  Laplacian eigenvalues}

For any block graph $G$, Bapat {\em et al.} \cite{BLP12} had shown  
that, under some conditions, the size of a maximal clique is a Laplacian eigenvalue of $G$. 
This fact was also observed in Abreu {\em et al.} \cite{Ab18} for the block-indifference graphs. 
For quasi-threshold graphs, we highlight (Theorem \ref{theo:quasi-structural}) the importance of the number of simplicial vertices in 
the determination of such eigenvalues.
In this section we generalize these results proving that, 
for any chordal graph $G$, the cardinality of all maximal cliques with at least two simplicial vertices
is an integer Laplacian eigenvalue of $G$;
their multiplicities are also determined.
The set of simplicial vertices of $Q_i$, $Simp(Q_i)$, plays a decisive role.

\begin{theorem}\label{theo:integer-eigen2}
Let $G=(V,E)$ be a connected chordal graph and let
 $\mathbb{Q }= \{ Q_j, 1 \leq j \leq \ell \}$, be the set of maximal cliques of $G$.
If $|Simp(Q_i)| \geq 2$, $1 \leq i \leq \ell$, then  $|Q_i|$ is an eigenvalue of $L(G)$ 
with multiplicity at least $|Simp(Q_i)|-1$.
\end{theorem}

\begin{proof}
Let $G$ be a non-complete chordal graph and $Q_i \in \mathbb{Q}$, $|Q_i| = n_i$. 
Let $u, v$ be two simplicial vertices belonging $Q_i$.
Since $G$ is a chordal graph and $u$ and $v$ are simplicial vertices in the same clique, they are true twins. 
So $N[u]=N[v]$ and, as $|Q_i| = n_i$,  then $d_G(u)=d_G(v)=n_i-1$. 
Considering  the entries of $L(G)$, we observe that: 

\begin{equation}\label{eq:laplacian-matrix}
L(G)_{u,u}=L(G)_{v,v} = d_G(u) = d_G(v)  \,\, \mbox{and} \,\, L(G)_{u,v}=L(G)_{v,u}=-1 
\end{equation}

Let ${\bf y}$ be the $n$-dimensional vector such that ${\bf y}(u)=1$, ${\bf y}(v)=-1$ and let the remaining entries be $0$.
From the values of entries of $L(G)$ given by  (\ref{eq:laplacian-matrix}), we conclude that $L(G) {\bf y} = n_i {\bf y}$. 
So, $n_i$ is an integer Laplacian eigenvalue of $G$.

For every $j$, $1 \leq j \leq \ell$, let $|Simp(Q_j)| = \{u_1^j, ...., u_{p_j}^j\}$.
Suppose $i \in \{1,\dots,\ell  \}$ such that  $p_i \geq 2$.
For all $r$, $2 \leq r \leq p_i$ build $p_i - 1$ vectors ${\bf y_r^i}$ such that 
${\bf y^i_r}(u_1^i)=1$, ${\bf y^i_r}(u_r^i)= - 1$ with all the remaining entries equal to zero.

The vectors ${\bf y^i_r}$,  $2 \leq r \leq p_i$,  are linearly independent. 
Besides, the vectors ${\bf y^i_r}$, $2 \leq r \leq p_i$,
satisfy $L(G){\bf y^i_r} = n_i {\bf y^i_r }$, so  they are eigenvectors of $L(G)$ corresponding to 
the eigenvalue $n_i$. 
Then,  the multiplicity of $n_i$ as Laplacian eigenvalue 
of $G$ is $p_i - 1$ and the result follows. 
\qed
\end{proof}

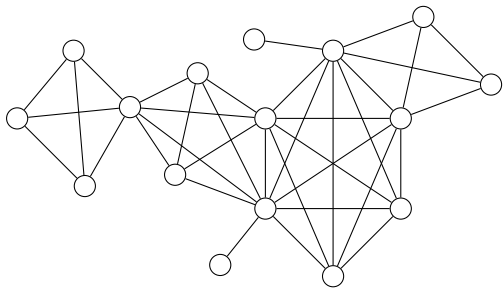
\begin{figure}
\begin{center}
\begin{tikzpicture}
  [scale=.3,auto=left]
\tikzstyle{every node}=[circle, draw, fill=white,
                         inner sep=0pt, minimum width=8pt]
\node (a1) at (-5,1){};
\node (b1) at (-8,4){};
\node (c1) at (-5.5,7){};
\node (a) at (6,-3){};
\node (b) at (3,0){};
\node (c) at (3,4){};
\node (d) at (6,7) {};
\node (e) at (9,4){};
\node (f) at (9,0){};
\node (x) at (0,6){};
\node (y) at (-3,4.5){};
\node (z) at (-1,1.5){};
\node (1) at (10,8.5){};
\node (2) at (13,5.5){};
\node (3) at (2.5,7.5){};
\node (4) at (1,-2.5){};
\foreach \from/\to in {a/b, a/c,a/d, a/e,a/f, b/d,b/c, b/e,b/f,c/d, c/e, c/f, d/e,d/f,e/f,  
x/y,x/z,y/z,x/b,x/c,y/b,y/c,z/b,z/c,
1/2,1/e,2/e,1/d,2/d,4/b,3/d,
a1/b1,a1/y,b1/y,a1/c1,b1/c1,c1/y} 
\draw (\from) -- (\to);
\end{tikzpicture}
\caption{A chordal graph with $\ell=6$}
\label{fig:simp-chordal}
\end{center}
\end{figure}

As an example, the Laplacian spectrum of $G$, presented in Figure \ref{fig:simp-chordal},  is:

$SpecL(G) = [-,00000, ,53913, ,86586, 1,10090, 1,79600, 3,06443, 4,00000,$

$ 4,00000,4,00000, 5,00000, 6,00000, 7,41877, 7,84776, 8,71138, 9,35720,$ 

$ 10,29857]$.

From Theorem \ref{theo:integer-eigen2}, the following  integer Laplacian eigenvalues of $G$ are directly determined: 
$4$ with multiplicity $3$, 
$5$ with multiplicity $1$ and  $6$   with multiplicity $1$.


\section{Conclusions}

Several subclasses of chordal graphs were mentionned in this paper.
Figure \ref{fig:hier}  shows how they relate.

\begin{figure}[h]
\begin{center}
\begin{tikzpicture}  [scale=.6]
\tikzstyle{block} = [rectangle, draw,  
text width=5.8em, text centered, rounded corners, minimum height=2em]
\tikzstyle{line} = [draw, -latex']   
    
\node[block] (1)      at (0,0)  {chordal};
\node[block] (1b)      at (8,0)  {cograph};
\node[block] (2a)    at (-3.5,-2) {split};
\node[block] (2b)    at (2,-2) {graphs with $a(G)=\kappa(G)$};
\node[block] (3a)    at (4.5,-4.5) {quasi-threshold};
\node[block] (3c)    at (-4.5,-7) {$(k,t)$-split};
\node[block] (5a) at (1,-7) {threshold};
\node[block] (5b)    at (7,-7) {gen core-sat};
\node[block] (5c)    at (3,-10) {split-complete};
\node[block] (6b)    at (10,-10) {windmill};
    
\draw (1)--(2a)--(5a);
\draw(2a)--(3c);
\draw  (1)--(2b)--(3a); 
 \draw (1b)--(3a)--(5a); 
 \draw (3c)--(5c);
 \draw (5a)--(5c);
 \draw (3a)--(5b); 
 \draw (5b)--(5c);
 \draw (5b)--(6b);
\end{tikzpicture}
\end{center}
\caption{A hierarchy of chordal graphs} 
\label{fig:hier}
\end{figure}
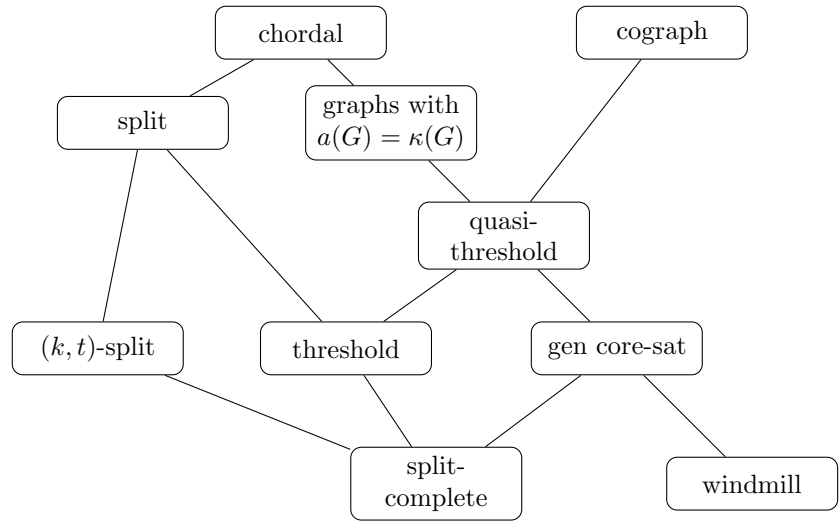

In all these subclasses 
we could notice the importance of the graph structure, its minimal vertex separators,
its simplicial vertices and the maximal cliques.
Other subclasses can be explored in the future.



\begin{thebibliography}{10}\label{bibliography}

\bibitem{Ab18} N.M.M. de Abreu, C.M.Justel, L. Markenzon, C.S.Oliveira, C.F.E.M.Waga, 
Block-indifference graphs: Characterization, structural and spectral properties, 
Discrete Appl. Math. (2018), https://doi.org/10.1016/j.dam.2018.11.034.


\bibitem{BLP08}    R. B. Bapat, R.B., Lal, A.K.,  Pati, S.,
Laplacian Spectrum of Weakly Quasi-threshold Graphs, 
Graphs Combin. 24(2008), pp 273–290.


\bibitem{BLP12} R.B. Bapat, A.K. Lal, S. Pati,
On algebraic connectivity of graphs with at most two points of
articulation in each block,
Linear Multilinear Algebra 60 (2012) 415--432.


\bibitem{BP93} J.R.S., Blair, B. Peyton,
An Introduction to Chordal Graphs and Clique Trees,
In Graph Theory and Sparse Matrix Computation, IMA 56, pp. 1--29, 1993.

\bibitem{BLS99} A. Brandst\"adt, V. B. Le, J. Spinrad,
Graph Classes - a Survey, SIAM Monographs in Discrete Mathematics and Applications,
Philadelphia, PA,1999.

\bibitem{EB17} E. Estrada, M. Benzi, 
Core-satellite graphs: Clustering, assortativity and spectral properties,
Linear Algebra Appl. 517 (2017) 30-52.

\bibitem{F73} M.M. Fiedler, 
Algebraic connectivity of graphs,
 Czvchoslovak Math. J. 23 (1973) 298-305.

\bibitem{F09}
M. A. A. de Freitas, 
Grafos integrais, laplacianos integrais e Q-integrais (in Portuguese), 
Ph. D. Thesis, UFRJ, Rio de Janeiro, 2009.


\bibitem{Go04} M.C. Golumbic, 
Algorithmic Graph Theory and Perfect Graphs, $2^{nd}$ edition,
Academic Press, New York, 2004.

\bibitem{G78} M.C. Golumbic, 
Trivially perfect graphs, 
Discrete Math. 24 (1978) 105-107.

\bibitem{K10} S. Kirkland, M.A.A. Freitas, R.R. Del Vecchio, N.M.M. Abreu,
Split non-threshold Laplacian integral graphs, 
Linear Algebra Appl. 58(2) (2010) 221-233.

\bibitem{K02} 
S.J. Kirkland, J.J. Molitierno, N. Newmann, B.L. Shader, 
On graphs with equal algebraic and vertex connectivity,
Linear Algebra Appl. 341 (2002) 45-56.

\bibitem{MWW89}
S. Ma, W.D. Wallis,  J. Wu, 
Optimization problems on quasi-threshold graphs, 
J. Combin. Inform. System. Sci. 14 (1989) 105-110.

\bibitem{MP10}  L. Markenzon, P.R.C.  Pereira,
One phase algorithm for the determination of minimal vertex separators of chordal graphs,
Int. Trans. Oper. Res. 17 (6) (2010) 683-690. 



\bibitem{W62} E.S. Wolk, 
The comparability graph of a tree, 
Proc. Amer. Math. Sot. 3 (1962) 789-795.

\bibitem{YCC96} J.-H. Yan, J.-J. Chen, G. J. Chang,
Quasi-threshold graphs,
Discrete Appl. Math. 69 (1996) 247-255.


\end{thebibliography}
\end{document}